\documentclass[hidelinks,11pt,english]{article} 
\usepackage{etex,etoolbox}
\usepackage[pdftex]{graphicx}
\graphicspath{ {./images/} }
\usepackage{ulem}
\usepackage{multirow}
\usepackage[dvipsnames]{xcolor}
\usepackage{tikz}
\usetikzlibrary{decorations.pathreplacing}
\usepackage[ruled]{algorithm2e}

\SetCommentSty{mycommfont}
\usepackage{pgfplots}
\pgfplotsset{compat=newest}
\usepackage{hyperref}
\usepackage{float}
\usepackage{amsmath, amsthm, mathrsfs, amssymb}
\usepackage{amsfonts}
\usepackage{subfigure}
\usepackage{mathrsfs}
\usepackage{datetime}
\usepackage{enumerate}
\usepackage{booktabs}
\usepackage{array}
\newcolumntype{H}{>{\setbox0=\hbox\bgroup}c<{\egroup}@{}}
\usepackage{amsthm}
\usepackage{amsmath}
\usepackage[round]{natbib}
\usepackage[margin=1in,bottom=0.9in,includefoot]{geometry}
\usepackage[utf8]{inputenc}
\usepackage{color}
\usepackage{verbatim}
\usepackage{sectsty}

\usepackage{chngcntr}
\usepackage{apptools}
\AtAppendix{\counterwithin{lemma}{section}}
\makeatletter
\providecommand{\@fourthoffour}[4]{#4}
\def\fixstatement#1{%
  \AtEndEnvironment{#1}{%
    \xdef\pat@label{\expandafter\expandafter\expandafter
      \@fourthoffour\csname#1\endcsname\space\@currentlabel}}}

\globtoksblk\prooftoks{1000}
\newcounter{proofcount}

\long\def\proofatend#1\endproofatend{%
  \edef\next{\noexpand\begin{proof}[Proof of \pat@label]}%
  \toks\numexpr\prooftoks+\value{proofcount}\relax=\expandafter{\next#1\end{proof}}
  \stepcounter{proofcount}}

\def\printproofs{%
  \count@=\z@
  \loop
    \the\toks\numexpr\prooftoks+\count@\relax
    \ifnum\count@<\value{proofcount}%
    \advance\count@\@ne
  \repeat}
\makeatother

\newtheorem{thm}{Theorem}
\newtheorem{proposition}{Proposition}
\newtheorem{lemma}{Lemma}
\newtheorem{corollary}{Corollary}

\theoremstyle{definition}
\newtheorem{defn}{Definition}
\newtheorem{exmp}{Example}
\newcommand{\argmax}{\operatornamewithlimits{argmax}}

\fixstatement{thm}
\fixstatement{lemma}
\fixstatement{proposition}
\fixstatement{corollary}

\definecolor{ALLIN}{HTML}{151515}
\definecolor{BORD}{HTML}{000000}
\definecolor{POSSIBLY}{HTML}{A4A4A4}
\definecolor{EMPTY}{HTML}{FFFFFF}

\begin{document}
\title{Assortment Optimization \\ under the Sequential Multinomial Logit Model}
\author{
Alvaro Flores\footnote{College of Engineering \& Computer Science, Australian National University (alvaro.flores@anu.edu.au).}
\and
Gerardo Berbeglia\footnote{Melbourne Business School, The University of Melbourne (g.berbeglia@mbs.edu).}
\and
Pascal Van Hentenryck\footnote{Industrial and Operations Engineering \& Computer Science and Engineering, University of Michigan, Ann Arbor (pvanhent@umich.edu).}
}

\maketitle

\begin{abstract}
We study the assortment optimization problem under the
  \emph{Sequential Multinomial Logit} (SML), a discrete choice model
  that generalizes the multinomial logit (MNL). Under the SML model,
  products are partitioned into two levels, to capture differences in
  attractiveness, brand awareness and, or visibility of the products in the market. When a
  consumer is presented with an assortment of products, she first considers
  products in the first level and, if none of them is purchased,
  products in the second level are considered. This model is a special case of the
  Perception-Adjusted Luce Model (PALM) recently proposed by
  \citet{echenique2018}. It can explain many behavioral phenomena such as the
  attraction, compromise, similarity effects and choice overload which cannot be
  explained by the MNL model or any discrete choice model based on
  random utility. In particular, the SML model allows violations to
   \emph{regularity} which states that the probability of
  choosing a product cannot increase if the offer set is enlarged.

  This paper shows that the seminal concept of revenue-ordered
  assortment sets, which contain an optimal assortment under the MNL model, can be
  generalized to the SML model. More precisely, the paper proves that all
  optimal assortments under the SML are revenue-ordered by level, a
  natural generalization of revenue-ordered assortments that contains, at most, a quadratic number of assortments. As a
  corollary, assortment optimization under the SML is polynomial-time
  solvable. This result is particularly interesting given that the SML model does not satisfy the regularity condition and, therefore,
  it can explain choice behaviours that cannot be explained by any choice model based on random utility.
\end{abstract}	

\section{Introduction}
\label{intro}

The assortment optimization problem is a central problem in revenue
management, where a firm wishes to offer a set of products with the
goal of maximizing the expected revenue. This problem has many
relevant applications in retail and revenue management
\citep{KokAssortment}. For example, a publisher might need to decide
the set of advertisements to show, an airline must decide which fare
classes to offer on each flight, and a retailer needs to decide which
products to show in a limited shelf space.

The first consumer demand models studied in revenue management were
based on the independent demand principle. This principle stated that
customers decide beforehand which product they want to purchase: If
the product is available, they make the purchase and, otherwise, they
leave without purchasing. In these models, the problem of finding the
best offer set of products is computationally simple, but this
simplicity comes with an important drawback: These models do not
capture the substitution effects between products. That is, they
cannot model the fact that, when a consumer cannot find her/his
preferred product, she/he may purchase a substitute product.  It is
well-known that choice models that incorporate substitution effects
improve demand predictions
\citep{talluri2004revenue,Newman2014,vanRyzin2014,
  Ratliff2008,Vulcano2010}. One of the most celebrated discrete choice
models is the \emph{Multinomial Logit} (MNL)
\citep{luce1959,McFadden1974}. Under the MNL model, the assortment
problem admits a polynomial-time algorithm \citep{talluri2004revenue}.
However, the model suffers from the independence of irrelevant
alternatives (IIA) property \citep{ben1985} which says that, when a
customer is asked to choose among a set of alternatives $S$, the ratio
between the probability of choosing a product $x \in S$ and the
probability of choosing $y \in S$ does not depend on the set $S$. In
practice, however, the IIA property is often violated.
To overcome this limitation, more complex
choice models have been proposed in the literature such as the Nested
Logit model \citep{williams1977formation}, the latent class MNL
\citep{GreeneLatent}, the Markov Chain model \citep{blanchet2016markov},
and the exponomial model \citep{daganzo1979multinomial,exponomial2016}. All these
models satisfy the following property: The probability of choosing an
alternative does not increase if the offer set is enlarged. Despite
the fact that this property (known as regularity) appears natural, it
is well-known that it is sometimes violated by individuals
\citep{Debreu1960,Tversky1972choice,Tversky1972elim,Tversky_reg}. Recently, there
have been efforts to develop discrete choice models that can explain complex choice behaviours such as the
violation of regularity, one of the most prominent examples is the perception-adjusted Luce model (PALM) \citep{echenique2018}.

While the PALM and the nested logit are both conceived as sequential
choice processes, they have important differences. Probably the most
 important difference is that the nested logit model belongs to the family of
random utility models (RUM)\footnote{This is unless nest specific parameters
	 are greater than one, a case rarely studied in the literature.},
and therefore can't accommodate regularity violations. On the other hand
the PALM does not belong to the RUM class, and allows regularity
violations as well as choice overload. In terms of the choice process,
 in the nested logit model customers first select a nest, and then a product within the nest.
In the perception-adjusted Luce's model products are separated by preference levels,
so when a customer is offered a set of products, she first chooses among the offered
products belonging to the lowest available level, and if none of them are chosen then
she selects among the next available level, and keeps repeating this process until no more levels
are available or until a purchase is made.

\subsection{Our Contributions}
\label{contribution}
In this paper, we study the assortment optimization problem for a
two-stage discrete choice model model that generalizes the classical
Multinomial Logit model. This model, which we call the
\emph{Sequential Multinomial Logit} (SML) for brevity, is a special
case of the recently proposed model known as the perception-adjusted
Luce model (PALM) \citep{echenique2018}.
In the SML model, products are partitioned \emph{a priori} into two sets,
which we call levels. This product segmentation
into two levels can capture different degree of attractiveness.
For example, it can model customers who
check promotions/special offers first before considering the purchase of regular-priced products.
It can also model consumer brand awareness, where customers first check products
of	specifics brands before considering the rest. Finally,
 the SML can model product visibilities in a market, where products are placed in
specific positions (aisles, shelves, web-pages, etc.) that
induce a sequential analysis, even when all the products are at sight.
Our main contribution is to provide a polynomial-time algorithm for the
assortment problem under the SML and to give a complete
characterization of the resulting optimal assortments.

A key feature of the PALM and the SML, is their ability to capture several
effects that cannot be explained by any choice model based in random
utility (such as for example the MNL, the mixed MNL, the markov chain model, and the stochastic preference
models). Examples of such effects include attraction,
\citep{Doyle1999}, the compromise effect \citep{Simonson1992}, the similarity effect
\citep{Debreu1960,Tversky1972elim}, and the paradox of choice (also known as choice overload)
\citep{iyengar2000choice,schwartz2004,Haynes2009,chernev2015choice}. These effects are discussed in the next
section. In particular, the SML allows for violations
of regularity. There are very few analyses of assortment problems under a choice model outside the RUM class.

Our algorithm is based on an in-depth analysis of the structure of the
SML. It exploits the concept of revenue-ordered assortments that
underlies the optimal algorithm for the assortment problem under the
MNL. The key idea in our algorithm is to consider an assortment built
from the union of two sets of products: A revenue-ordered assortment
from the first level and another revenue-ordered assortment from the
second level. Several structural properties of optimal assortments
under the SML are also presented.

\subsection{Relevant Literature}
\label{relevant_literature}
The heuristic of revenue-ordered assortments,
consists in evaluating the expected revenue of all the assortments that
can be constructed as follows: fix threshold revenue $\rho$ and then
select all the products with revenue of at least $\rho$.
This strategy is appealing because it can be applied to assortment
problems for any discrete choice model. In addition, it
only needs to evaluate as many assortments as there are different
revenues among products. In a seminal result,
\cite{talluri2004revenue} showed that, under the MNL model, the
optimal assortment is revenue-ordered. This result does not hold for
all assortment problems however. For example, revenue-ordered
assortments are not necessarily optimal under the MNL model with
capacity constraints
\citep{rusmevichientong2010dynamic}. Nevertheless, in another seminal
result, \cite{rusmevichientong2010dynamic} showed that the assortment
problem can still be solved optimally in polynomial time under such
setting.

\cite{RobustMNL2012} considered a model where customers make choices
following an MNL model, but the parameters of this model belong to a
compact uncertainty set, i.e., they are not fully determined. The firm
wants to be protected against the worst-case scenario and the problem
is to find an optimal assortment under these uncertainty
conditions. Surprisingly, when there is no capacity constraint, the
revenue-ordered strategy is optimal in this setting as well.

There are also studies on how to solve the assortment problem when
customers follow a mixed multinomial logit
model. \cite{bront2009column} showed that this problem is NP-hard in the
strong sense using a reduction from the \textit{minimum vertex cover problem}
\citep{Garey1979}. \cite{MendezDiaz2014} proposed a
branch-and-cut algorithm to solve the optimal assortment under the
Mixed-Logit Model. An algorithm to obtain an upper bound of the
revenue of an optimal assortment solution under this choice model was
proposed by \cite{FeldmanBounding}. \cite{rusmevichientong2014} showed
that the problem remains NP-hard even when there are only two
customers classes.
			
Another model that attracted researchers attention is the
\textit{nested logit model} \citep{williams1977formation}.
Under the nested logit model, products are partitioned into nests,
and the selection process for a customer goes by first selecting a nest,
and then a product within the selectednest.
It also has a dissimilarity parameter associated with each nest
that serves the purpose of magnifying or dampening the total
preference weight of the nest. For the two-level nested logit model,
\cite{davis2014assortment} studied the assortment problem and showed
that, when the dissimilarity parameters are bounded by $1$ and the
no-purchase option is contained on a nest of its own, an optimal
assortment can be found in polynomial time; If any of these two
conditions is relaxed, the resulting problem becomes NP-hard, using a
reduction from the \textit{partition problem} \citep{Garey1979}. The
polynomial-time solution was further extended by
\cite{gallego2014nested}, who showed that, even if there is a capacity
constraint per nest, the problem remains polynomial-time solvable.
\cite{Li2015nested} extended this result to a d-level nested logit
model (both results under the same assumptions over the dissimilarity
parameters and the no-purchase option). \cite{Jagabathula} proposed a
local-search heuristic for the assortment problem under an arbitrary
discrete choice model. This heuristic is optimal in the case of the
MNL, even with a capacity constraint.

\cite{wang2017impact} has studied the assortment optimization in a context in which consumer
search costs are non-negligible. The authors showed that the strategy of revenue-ordered assortments is not optimal.
Another interesting model sharing similar choice probabilities
to those of the PALM, is the one proposed in
\cite{manzini2014stochastic} which is based on consider first and choose second process.
\citet{echenique2018} showed that the PALM and the model by Manzini and Mariotti are in fact disjoint.
The assortment optimization problem under the Manzini and Mariotti model has been recently
studied by \citet{gallego2017attention}, where they show that revenue-ordered assortments strategy is optimal.
Another choice model studied is the negative exponential distribution (NED) model \citep{daganzo1979multinomial},
also known as the Exponomial model \citep{exponomial2016} in which customer utilities follow negatively skewed distribution.
\citet{exponomial2016} proved that when prices are exogenous, the optimal assortment might not be revenue-ordered assortment,
because a product can be skipped in favour of a lower-priced one
depending on the utilities. This last result differs from what happens under the MNL and the Nested Logit Model (within each nest).
Another recently proposed extension to the multinomial logit model, is the General Luce Model (GLM) \citep{Gen_Luce}.
The GLM also generalizes the MNL and falls outside the RUM class. In the GLM, each product has an
intrinsic utility and the choice probability depends upon a dominance relationship between the products.
Given an assortment $S$, consumers first discard all
dominated products in $S$ and then select a product from the remaining
ones using the standard MNL model. \citet{2017GMNLArxiv_unpublished} studied the assortment optimization problem under the GLM.
The authors showed that revenue ordered assortments are not optimal, and proved that the problem can still be solved in polynomial time.

Recently, \cite{BerbegliaRev} studied how well revenue-ordered assortments
approximate the optimal revenue for a large class of choice models,
namely all choice models that satisfy regularity. They provide three types of revenue guarantees
that are exactly tight even for the special case of the RUM family.
In the last few years, there has been progress in studying the assortment problem
in choice models that incorporate visibility or position biases.
In these models, the likelihood of selecting an alternative not only depends on the offer set,
but also in the specific positions at which each product
is displayed \citep{abeliuk2015assortment, aouad2015display, davis2013assortment, gallego2016approximation}.

We mentioned that PALM and SML are able to accommodate many effects that can't be
explained by models the RUM class. We briefly describe each one of them in the
following paragraphs.

The attraction effect stipulates that, under certain conditions,
adding a product to an existing assortment can increase the
probability of choosing a product in the original assortment. We
briefly describe two experiments of this effect. \cite{Simonson1992}
considered a choice among three microwaves $x, y,$ and $z$. Microwave
$y$ is a Panasonic oven, perceived as a good quality product, and $z$
is a more expensive version of $y$. Product $x$ is an Emerson
microwave oven, perceived as a lower quality product. The authors
asked a set of 60 individuals ($N=60$) to choose between $x$ and $y$;
they also asked another set of 60 participants ($N=60$) to choose among
$x, y,$ and $z$. They found out that the probability of choosing $y$
increases when product $z$ is shown. This is a direct violation of
regularity, which states that the probability of choosing a product
does not increase when the choice set is enlarged, as described by
\cite{McCausland2013}. Another demonstration of the attraction effect
was carried by \cite{Doyle1999} who analyzed the choice behaviour of
two sets of participants ($N=70$ and $N=82$) inside a grocery store in
the UK varying the choice set of baked beans. To the first group, they
showed two types of baked beans: Heinz baked beans and a local (and
cheaper) brand called Spar. Under this setting, the Spar beans was
chosen $19\%$ of the time. To the second group, the authors introduced
a third option: a more expensive version of the local brand
Spar. After adding this new option, the cheap Spar baked beans was
chosen $33\%$ of the time. It is worth highlighting that the choice
behaviour in these two experiments cannot be explained by a
Multinomial Logit Model, nor can it be explained by any choice model based on random utility.

The compromise effect \citep{Simonson1992} captures the fact that
individuals are averse to extremes, which helps products that
represent a ``compromise'' over more extreme options (either in price,
familiarity, quality, ...). As a result, adding extreme options
sometimes leads to a positive effect for compromise products, whose
probabilities of being selected increase in relative terms compared to
other products. This phenomenon violates again the IIA axiom of
Luce's model and the regularity axiom satisfied by all random utility models \citep{BerbegliaRev}.
Again, the compromise effect can be captured with the PALM.

The similarity effect is discussed in \cite{Tversky1972elim},
elaborating on an example presented in \cite{Debreu1960}: Consider $x$
and $z$ to represent two recordings of the same Beethoven symphony and
$y$ to be a suite by Debussy. The intuition behind the effect is that
$x$ and $z$ jointly compete against $y$, rather than being separate
individual alternatives. As a result, the ratio between the
probability of choosing $x$ and the probability of choosing $y$ when
the customer is shown the set $\left\{x,y\right\}$ is larger than the
same ratio when the customer is shown the set $\left\{x,y,z\right\}$.
Intuitively, $z$ takes a market share of product $x$, rather than a
market share of product $y$.

Finally, the choice overload effect occurs when the probability of making
a purchase decreases when the assortment of available products is enlarged.
To our knowledge, the first paper that shows the empirical existence of choice
overload is written \citet{iyengar2000choice}.
In their experimental setup, customers are offered jams from a
tasting booth displaying either $6$ (limited selection)
or 24 (extensive selection) different flavours.
All customers were given a discount coupon for making a purchase of one of the offered jams.
Surprisingly, $30\%$ of the customers
offered the \textit{limited selection} used the coupon, while only $3\%$
of customers offered the \textit{extensive selection} condition used the coupon.
Another studies of choice overload are in 401(k) plans \cite{iyengar2004much},
chocolates \citep{chernev2003more}, consumer electronics \citep{chernev2003product} and  pens \citep{shah2007buying}.
For a more in depth discussion of this effect the reader is referred to \cite{schwartz2004}.
Readers are also referred to \cite{chernev2015choice} for a review and meta-analysis on this topic.

\section{Problem Formulation}
\label{PF}

This section presents the sequential multinomial logit model
considered in this paper and its associated assortment optimization
problem. Let $X$ be the set of all products and $x_0$ be the no-choice
option ($x_0 \notin X)$. Following \cite{echenique2018}, each product $x \in X$
is associated with an intrinsic utility $u(i)> 0$ and a perception
priority level $l(x) \in \{1,2\}$. The idea is that customers perceive
products sequentially, first those with priority 1 and then those with
priority 2. This perception priority order could represent differences
in familiarity, degree of attractiveness or salience of different products, or even in the
way the products are presented. Let $X_i=\{x \in X: l(x)=i\}$ be the
set of all products belonging to level $i=1,2$. Given an assortment $S
\subseteq X$, we write $S=S_1 \uplus S_2$ with $S_1 \subseteq X_1$ and
$S_2 \subseteq X_2$ to denote the fact that $S$ is a partition
consisting of two subsets $S_1$ and $S_2$.

The \emph{Sequential Multinomial Logit Model} (\textbf{SML}) is a
discrete choice model where the probability $\rho(x,S)$ of choosing a
product $x$ in an assortment $S$ is given by:
\[
\rho(x,S)= \left\{
\begin{array}{ll}
\frac{u(x)}{\sum_{y \in S}u(y)+u_0} & \text{if } x\in S_1, \\
\left[1-\frac{\sum_{z \in S_1}u(z)}{\sum_{y \in S}u(y)+u_0}\right]\cdot\frac{u(x)}{\sum_{y \in S}u(y)+u_0} & \text{if } x\in S_2.
\end{array}
\right.
\]
\noindent
where $u_0$ denotes the intrinsic utility of the no-choice option,
which has a probability
\[
\rho(x_0,S)=1-\sum_{i \in S}\rho(i,S)
\]
of being chosen.

 Observe that the probability of
choosing a product $x \in S_1$ (which implies that $l(x)=1$ and $x \in
S$) is given by the standard MNL formula, whereas the probability of
choosing a product $y$ that belongs to the second level is given by
the probability of not choosing any product belonging to the level 1
multiplied by the probability of selecting product $y$ according the
MNL again.  Note that, if all the offered products belong to the same
level, this model is equivalent to the classical MNL model. The SML corresponds
to PALM restricted to two levels. We provide a full
description of PALM in Appendix \ref{App:PALM}.

Let $r:X\cup \{x_0\} \rightarrow \mathbb{R^+}$ be the revenue function
which assigns a per-unit revenue to each product and let
$r(x_0)=0$.  We use $R(S)$ to denote the \emph{expected revenue}
of an assortment $S$, i.e.,
	
\begin{equation}\label{Expected_Revenue}
R(S)=\sum_{x \in S}\rho (x,S)\cdot r(x).
\end{equation}

\noindent
The assortment optimization problem under the SML consists in finding an assortment $S^*$ that maximizes $R$, i.e.,
\begin{equation}
\label{Assortment_Opt-arg}
S^*=\argmax\limits_{S\subseteq X}R(S).
\end{equation}
We use $R^*$ to denote the maximum expected revenue, i.e.,
\begin{equation}
\label{Assortment_Opt}
R^*=\max\limits_{S\subseteq X}R(S).
\end{equation}

\noindent
Without loss of generality, we assume that $u(i)> 0$ in the rest of this
paper. We use $x_{ij}$ to denote the $j^{th}$ product of the $i^{th}$
level ($i=1,2$), and $m_i$ to denote the number of products in level $i$.
 Also, we assume that the products in each level are
indexed in a decreasing order by revenue (breaking ties arbitrarily), i.e.,
\[
\forall i \in \left\{1,2\right\}, r(x_{i1})\geq r(x_{i2})\geq\ldots\geq r(x_{im_{i}}).
\]

\noindent
It is useful to illustrate how the SML allows for violations of the
regularity condition, a property first observed by \cite{echenique2018}. Our
first example captures the attraction effect presented earlier.

\begin{exmp}[Attraction Effect in the SML]
  Consider a retail store that offers different brands of chocolate.
  Suppose that there is a well-known brand A and the brand B owned by
  the retail store. There is one chocolate bar $a_1$ from brand A and
  there are two chocolate bars $b_1$ and $b_2$ from Brand B, with
  $b_2$ being a more expensive version of $b_1$. When shown the
  assortment $\{a_1,b_1\}$, 71\% of the clients purchase $a_1$ and
  8.2\% buy $b_1$. When shown the assortment $\{a_1,b_1,b_2\}$,
  customers select $a_1$ 49.8\% of the time and, surprisingly,
  bar $b_1$ increases its market share to about 10\%, while
  bar $b_2$ accounts for 15\% of the market. The introduction of $b_2$
  to the assortment increases the purchasing probability of $b_1$,
  violating regularity. The numerical example can be explained with
  the SML as follows: Consider $A=\left\{a_1\right\}$,
  $B=\left\{b_1,b_2\right\}$ and $X=A \uplus B$. With
  $u(a_1)=100,u(b_1)=40,u(b_2)=60$, and $u_0=1$ as the utility of the
  outside option, we have:
  \begin{equation*}
	\rho(b_1,\left\{a_1,b_1\right\})=\frac{40}{141}\cdot\left[1-\frac{100}{141}\right]\approx 8.2\%.
  \end{equation*}
  and
  \begin{equation*}
  \rho(b_1,\left\{a_1,b_1,b_2\right\})=\frac{40}{201}\cdot\left[1-\frac{100}{201}\right]\approx 10\%.
  \end{equation*}	

\noindent		
 Hence $\rho(b_1,\left\{a_1,b_1\right\})<\rho(b_1,\left\{a_1,b_1,b_2\right\})$ which contradicts regularity.
\end{exmp}

Our second example shows that the SML can capture the so-called
\textit{paradox of choice} or \textit{choice overload effect} (e.g., \cite{schwartz2004,chernev2015choice}): The overall
purchasing probability may decrease when the assortment is enlarged.
Once again, this effect cannot be explained by any random utility
model and it is sometimes called the effect of ``too much choice''.

\begin{exmp}[Paradox of Choice in the SML]
  Let $X_1=\left\{x_{11}\right\}$, $X_2=\left\{x_{21},x_{22}\right\}$,
  $X=X_1\uplus X_2$, $u(x_{11})=10$ ,$u(x_{21})=1$
  ,$u(x_{22})=10$, and $u_0=1$. We have
\begin{align*}
\rho(x_0,\left\{x_{11},x_{21}\right\})&=1-\rho(x_{11},\left\{x_{11},x_{21}\right\})-\rho(x_{21},\left\{x_{11},x_{21}\right\})\\
&=1-\frac{10}{12}-\left(1-\frac{10}{12}\right)\cdot\frac{1}{12}=0.152\overline{7},
\end{align*}
and
\begin{align*}
\rho(x_0,\left\{x_{11},x_{21},x_{22}\right\})&=1-\rho(x_{11},\left\{x_{11},x_{21},x_{22}\right\})-\rho(x_{21},\left\{x_{11},x_{21},x_{22}\right\})-\rho(x_{22},\left\{x_{11},x_{21},x_{22}\right\})\\
&=1-\frac{10}{22}-\left(1-\frac{10}{22}\right)\cdot\frac{1}{22}-\left(1-\frac{10}{22}\right)\cdot\frac{10}{22}=0.\overline{27}.
\end{align*}
Hence $\rho(x_0,\left\{x_{11},x_{21}\right\})<\rho(x_0,\left\{x_{11},x_{21},x_{22}\right\})$.
\end{exmp}

In the following section we first focus in finding properties that any optimal solution of the assortment problem for the SML must satisfy. Then, in Section \ref{sec:optimality}, we use those properties to show the optimality of an extension of the classical revenue-ordered assortments which we called revenue-ordered assortments by level.

\section{Properties of Optimal Assortments}
\label{sec:properties-2-level}

In this section we derive properties of the optimal solutions to the assortment problem under the SML. These properties are extensively used in the proof of our main result (Theorem \ref{main_thm}) in Section \ref{sec:optimality}.  We establish bounds on the following: any product offered on any optimal solution, and also the assortments considered on an optimal solution on each level.  We assume a set of products $X = X_1 \uplus X_2$
and use the following notations
	
\begin{equation}
\label{defs}
U(S)=\sum_{x\in S}u(x),\quad \alpha(S)=\frac{\sum_{x\in S} u(x)r(x)}{\sum_{x\in S}u(x)}=\frac{\sum_{x\in S} u(x)r(x)}{U(S)}\quad \text{ and }\quad \lambda(Z,S)=\frac{U(Z)}{U(S)+u_0}
\end{equation}

\noindent where $Z \subseteq S$ and $Z,S \subseteq X$. Note that
$\alpha(S)$ is the usual MNL formula for the revenue and, when $S=\left\{x\right\}$ for some
$x\in X$, $\alpha(S)=r(x)$. With these notations, the revenue of an
assortment $S=S_1\uplus S_2$ is
	
\begin{align}
\label{rev_short}
R(S)&=\frac{\alpha(S_1)U(S_1)}{U(S_1)+U(S_2)+u_0}+\frac{\alpha(S_2)U(S_2)}{U(S_1)+U(S_2)+u_0}\cdot\left(1-\frac{U(S_1)}{U(S_1)+U(S_2)+u_0}\right)\nonumber\\
&=\frac{\alpha(S_1)U(S_1)}{U(S)+u_0}+\frac{\alpha(S_2)U(S_2)}{U(S)+u_0}\cdot\left(1-\frac{U(S_1)}{U(S)+u_0}\right).
\end{align}

\noindent
The following proposition is useful to divide a set into disjoint sets, which can then be analyzed separately.
	
\begin{proposition}
\label{decomposition}
Let $S\subseteq X$ and $S=H\cup T$ with $H\cap T=\emptyset$. We have
		
\begin{equation}\label{decomp}
\alpha(S)=\frac{\alpha(H)U(H)+\alpha(T)U(T)}{U(S)}.
\end{equation}
\end{proposition}

\begin{proof}
\[
\begin{aligned}
\frac{\alpha(H)U(H)+\alpha(T)U(T)}{U(S)}&=\frac{\frac{\sum_{x\in H}r(x)u(x)}{U(H)}\cdot U(H)+\frac{\sum_{x\in T}r(x)u(x)}{U(T)}\cdot U(T)}{U(S)}&& \text{/using definition of $\alpha(\cdot)$}\\
&=\frac{\sum_{x\in H}r(x)u(x)+\sum_{x\in T}r(x)u(x)}{U(S)}&& \text{/cancelling $U(H)$ and $U(T)$}\\
&=\frac{\sum_{x\in S}r(x)u(x)}{U(S)}&& \text{/using that $H\cup T=S$}\\
&=\alpha(S). &&\text{/definition of $\alpha(S)$}
\end{aligned}
\]
\end{proof}

\noindent
The next proposition is useful to bound expected revenues.

\begin{proposition}\label{bound_alpha}
Let $S_1,S_2\subseteq X$. If $\forall x\in S_1,\forall y\in S_2, r(x) \geq r(y)$, then $\alpha(S_1)\geq \alpha(S_2)$.
\end{proposition}
\begin{proof}
If $\forall x\in S_1,\forall y\in S_2: r(x)\geq r(y)$, then $\min_{x\in S_1}r(x)\geq \max_{y\in S_2}r(y)$. We have
\begin{equation*}
\alpha(S_1)=\frac{\sum_{x\in S_1} u(x)r(x)}{\sum_{x\in S_1}u(x)}\geq \min_{x\in S_1}r(x)\cdot\underbrace{\frac{\sum_{x\in S_1} u(x)}{\sum_{x\in S_1}u(x)}}_{1}\geq \max_{y\in S_2}r(y)\underbrace{\frac{\sum_{y\in S_2} u(y)}{\sum_{y\in S_2}u(y)}}_{1}\geq \alpha(S_2).
\end{equation*}
\end{proof}

\noindent
The following proposition bounds the MNL revenue of the products in the
first level. We use $S_i^*=S^*\cap X_i$ to denote the products in
level $i$ in the optimal assortment, i.e., $S^*=S_1^*\uplus S_2^*$.
	
\begin{proposition}[Bounding Level 1]
\label{R_bounds_1}
$\alpha(S_1^*)\geq R^*$.
\end{proposition}

\begin{proof}
The proof shows that the optimal revenue is a convex combination of
$\alpha(S_1^*)$ and another term by using Equation \eqref{rev_short}
and multiplying/dividing the revenue associated with the second
level by $\frac{U(S_2^*)+u_0}{U(S_2^*)+u_0}$. We have
\begin{align*}
R^*=&\frac{\alpha(S_1^*)U(S_1^*)}{U(S^*)+u_0}+\frac{\alpha(S_2^*)U(S_2^*)}{U(S^*)+u_0}\cdot\left(1-\frac{U(S_1^*)}{U(S^*)+u_0}\right)\\
=&\frac{\alpha(S_1^*)U(S_1^*)}{U(S^*)+u_0}+\frac{\alpha(S_2^*)U(S_2^*)}{U(S_2^*)+u_0}\cdot\underbrace{\frac{U(S_2^*)+u_0}{U(S^*)+u_0}}_{\left(1-\lambda(S_1^*,S^*)\right)\in (0,1)}\cdot\left(1-\frac{U(S_1^*)}{U(S^*)+u_0}\right)\\
=&\alpha(S_1^*)\lambda(S_1^*,S^*)+R(S_2^*)\left(1-\lambda(S_1^*,S^*)\right)^2.
\end{align*}
$R^*$ is a convex combination of $\alpha(S_1^*)$ and $R(S_2^*)(1-\lambda(S_1^*,S^*))$. By optimality of $R^*$, $R(S_2^*)\leq R^*$ and hence $\alpha(S_1^*)\geq R^*$.
\end{proof}
	
\noindent
We now prove a stronger lower bound for the value $\alpha(S^*_2)$ of the second level.
	
\begin{proposition}(Bounding Level 2)
\label{R2_bound}
$\alpha(S_2^*)\geq \frac{R^*}{1-\lambda(S_1^*,S^*)}$.	
\end{proposition}
\begin{proof}
The proof is similar to the one in Proposition \ref{R_bounds_1}.
\begin{align*}
R^*=&\frac{\alpha(S_1^*)U(S_1^*)}{U(S^*)+u_0}+\frac{\alpha(S_2^*)U(S_2^*)}{U(S^*)+u_0}\cdot\left(1-\frac{U(S_1^*)}{U(S^*)+u_0}\right)\\
=&\frac{\alpha(S_1^*)U(S_1^*)}{U(S_1^*)+u_0}\cdot\underbrace{\frac{U(S_1^*)+u_0}{U(S^*)+u_0}}_{1-\lambda(S_2^*,S^*)}+\alpha(S_2^*)\cdot\underbrace{\frac{U(S_2^*)}{U(S^*)+u_0}}_{\lambda(S_2^*,S^*)}\cdot\underbrace{\left(1-\frac{U(S_1^*)}{U(S^*)+u_0}\right)}_{1-\lambda(S_1^*,S^*)}\\
=&R(S_1^*)\cdot(1-\lambda(S_2^*,S^*)) + (\alpha(S_2^*)(1-\lambda(S_1^*,S^*)))\cdot\lambda(S_2^*,S^*).
\end{align*}
$R^*$ is a convex combination of and $R(S_1^*)$ and
$\alpha(S_2^*)(1-\lambda(S_1^*,S^*))$. By optimality of $R^*$,
$R(S_1^*)\leq R^*$ and $\alpha(S_2^*)\geq
\frac{R^*}{(1-\lambda(S_1^*,S^*))}$.
\end{proof}

\noindent
	
The following example shows that it not always the case that the inequality proved above holds if one considers the products in $S_2^*$ separately. That is, the inequality $r(x)\geq \frac{R^*}{1-\lambda(S_1^*,S^*)}$ for all $x \in S_2^*$ is not always true. 

\begin{exmp}\label{ex:lowerboundcounter}

	Let $X_1=\left\{x_{11}\right\}$,
	$X_2=\left\{x_{21},x_{22}\right\}$, and $X=X_1\uplus X_2$. Let
	the revenues be $r(x_{11})=10,r(x_{21})=9$, and $r(x_{22})=6$ and the
	utilities be $u(x_{11})=u(x_{21})=1, u(x_{22})=3$, and $u_0=1$.
	The expected revenue for all possible subsets are given by \\
	
	\begin{center}
		\begin{tabular}{@{} cc @{}}
			\toprule
			$S$ &        $R(S)$  \\
			\midrule
			$\{x_{11}\}$  & 5    \\	
			$\{x_{21}\}$ & 4.5   \\
			$\{x_{22}\}$  & 4.5    \\
			$\{x_{11}, x_{21}\}$ & $5.\overline{3}$    \\
			$\{x_{11}, x_{22}\}$ & 4.88   \\
			$\{x_{21}, x_{22}\}$  & 5.4    \\
			$\{x_{11},x_{21}, x_{22}\}$  & $5.41\overline{6}$   \\
			
			\bottomrule
		\end{tabular}
	\end{center}
	
	\noindent
	The optimal assortment is $S^*=\left\{x_{11},x_{21},x_{22}\right\}$
	with an expected revenue of $R^*=5.41\overline{6}$. By
	definition of $\lambda(\cdot)$, we have
	\[
	\lambda(S_1^*,S^*)=\frac{U(S_1^*)}{U(S^*)+u_0}=\frac{1}{6}=0.1\overline{6}.
	\]
	It follows that
	$r(x_{22})=6<\frac{R^*}{1-\lambda(S_1^*,S^*)}=\frac{5.41\overline{6}}{1-0.1\overline{6}}=6.488$, showing that the bound does not hold for product $x_{22}$.
\end{exmp}

However, the weaker bound holds for every product, and more generally,
we have the following proposition (a proof is provided in Appendix \ref{proofs}).

\begin{proposition}\label{global_bound}
	In every optimal assortment $S^*$, if $Z\subseteq S_i^*$ ($i=1,2$), then $\alpha(Z)\geq R^*$.
\end{proposition}

\proofatend
The proof of this proposition relies on the following technical lemma.

\begin{lemma}
	\label{convex_lemma}
	Consider an assortment $S=S_1\uplus S_2$ and $Z\subseteq S_i$ for some $i=1,2$. $R(S)$ can be expressed in terms of the following convex combinations:
	\begin{itemize}
		\item If $Z\subseteq S_1$,

		\begin{align}
		R(S)=& R(S\setminus Z)\cdot(1-\lambda(Z,S)) \nonumber\\ &+\left[\alpha(Z)-\frac{\alpha(S_2)U(S_2)(U(S_2)+u_0)(1-\lambda(Z,S)}{(U(S)-U(Z)+u_0)^2}\right]\cdot\lambda(Z,S).
		\end{align}
		
		\item if $Z\subseteq S_2$,
		
		\begin{align}
		R(S)=& R(S\setminus Z)\cdot\left(1-\lambda(Z,S)\right)\nonumber\\
		& +\left[\frac{\alpha(Z)(U(S_2)+u_0)}{U(S)+u_0}+\frac{\alpha(S_2\setminus Z)(U(S_2)-U(Z))}{U(S)-U(Z)+u_0}\cdot\frac{U(S_1)}{U(S)+u_0}\right]\cdot\lambda(Z,S).
		\end{align}
		
	\end{itemize}
	
\end{lemma}
\begin{proof}
	For $Z\subseteq X$ and $Z_0\subseteq Z$, we have:
	\begin{align}
	\alpha(Z\setminus Z_0)&=\frac{\sum_{x\in Z\setminus Z_0} r(x)u(x)}{U(Z\setminus Z_0)}\nonumber\\
	&=\frac{\sum_{x\in Z} r(x)u(x)}{U(Z\setminus Z_0)}-\frac{\alpha(Z_0)U(Z_0)}{U(Z\setminus Z_0)}\nonumber\\
	&=\frac{\sum_{x\in Z} r(x)u(x)}{U(Z\setminus Z_0)}\cdot\frac{U(Z)}{U(Z)}-\frac{\alpha(Z_0)U(Z_0)}{U(Z\setminus Z_0)}\nonumber\\
	&=\underbrace{\frac{\sum_{x\in Z} r(x)u(x)}{U(Z)}}_{\alpha(Z)}\cdot\frac{U(Z)}{U(Z)-U(Z_0)}-\frac{\alpha(Z_0)U(Z_0)}{U(Z)-U(Z_0)}\nonumber\\
	&=\frac{\alpha(Z)U(Z)-\alpha(Z_0)U(Z_0)}{U(Z)-U(Z_0)}\nonumber\\
	&=\frac{\alpha(Z)U(Z)-\alpha(Z_0)U(Z_0)}{U(Z)-U(Z_0)},
	\end{align}
	which can be rewritten as
	\begin{equation}\label{identity}
	\alpha(Z\setminus Z_0)(U(Z)-U(Z_0))=\alpha(Z)U(Z)-\alpha(Z_0)U(Z_0).
	\end{equation}
	Note also that, when $Z_0=Z$, $\alpha(Z\setminus Z_0)=\alpha(\emptyset)=0$.
	
	The rest of the proof is by case analysis on the level. If $Z\subseteq S_1$, $\lambda(Z,S)=\frac{U(Z)}{U(S)+u_0}$. We have:
	\small
	\begin{align*}
	R(S)=&\frac{\alpha(S_1)U(S_1)}{U(S)+u_0}+\frac{\alpha(S_2)U(S_2)}{U(S)+u_0}\cdot\left(1-\frac{U(S_1)}{U(S)+u_0}\right)\\
	=&\frac{\alpha(S_1)U(S_1)-\alpha(Z)U(Z)}{U(S)+u_0} +\frac{\alpha(Z)U(Z)}{U(S)+u_0}+\frac{\alpha(S_2)U(S_2)}{U(S)+u_0}\cdot\left(1-\frac{U(S_1)}{U(S)+u_0}\right)\\
	=&\frac{\alpha(S_1)U(S_1)-\alpha(Z)U(Z)}{U(S)+u_0}\cdot\frac{U(S)-U(Z)+u_0}{U(S)-U(Z)+u_0} +\frac{\alpha(Z)U(Z)}{U(S)+u_0}+\frac{\alpha(S_2)U(S_2)}{U(S)+u_0}\cdot\left(1-\frac{U(S_1)}{U(S)+u_0}\right)\\
	=&\frac{\alpha(S_1)U(S_1)-\alpha(Z)U(Z)}{U(S)-U(Z)+u_0}\cdot\left(1-\lambda(Z,S)\right) +\alpha(Z)\lambda(Z,S)\nonumber\\
	&+\frac{\alpha(S_2)U(S_2)(U(S_2)+u_0)}{(U(S)-U(Z)+u_0)^2}\cdot\left(\frac{U(S)-U(Z)+u_0}{U(S)+u_0}\right)^2,
	\end{align*}
	\normalsize
	
	\noindent
	where we first add and subtract $\frac{\alpha(Z)U(Z)}{U(S)+u_0}$,
	and multiply and divide the first term by $(U(S)-U(Z)+u_0)$. The
	last step uses the definition of $\lambda(Z,S)$ and multiplies
	and divides the last term by $(U(S)-U(Z)+u_0)^2$. Now applying
	Equation \eqref{identity} to $S_1$ and $Z$ in the
	last equation, we obtain
	
	\small
	\begin{align*}
	R(S)
	=&\frac{\alpha(S_1\setminus Z)(U(S_1)-U(Z))}{U(S)-U(Z)+u_0}\cdot\left(1-\lambda(Z,S)\right) +\alpha(Z)\lambda(Z,S)+\frac{\alpha(S_2)U(S_2)(U(S_2)+u_0)}{(U(S)-U(Z)+u_0)^2}\cdot\left(1-\lambda(Z,S)\right)^2\\
	=&\underbrace{\left[\frac{\alpha(S_1\setminus Z)(U(S_1)-U(Z))}{U(S)-U(Z)+u_0}+\frac{\alpha(S_2)U(S_2)(U(S_2)+u_0)}{(U(S)-U(Z)+u_0)^2}\right]}_{R(S_1\setminus Z\cup S_2)}\cdot\left(1-\lambda(Z,S)\right)\\ &+\left[\alpha(Z)-\frac{\alpha(S_2)U(S_2)(U(S_2)+u_0)\left(1-\lambda(Z,S)\right)}{(U(S)-U(Z)+u_0)^2}\right]\cdot\lambda(Z,S)\\
	=& R(S_1\setminus Z\cup S_2)\left(1-\lambda(Z,S)\right) +\left[\alpha(Z)-\frac{\alpha(S_2)U(S_2)(U(S_2)+u_0)\left(1-\lambda(Z,S)\right)}{(U(S)-U(Z)+u_0)^2}\right]\cdot\lambda(Z,S)\\
	=& R(S\setminus Z)\left(1-\lambda(Z,S)\right) +\left[\alpha(Z)-\frac{\alpha(S_2)U(S_2)(U(S_2)+u_0)\left(1-\lambda(Z,S)\right)}{(U(S)-U(Z)+u_0)^2}\right]\cdot\lambda(Z,S).
	\end{align*}
	\normalsize
	
	
	If $Z\subseteq S_2$, the proof is essentially similar. It also uses $\lambda(Z,S)=\frac{U(Z)}{U(S)+u_0}$ and apply Equation \eqref{identity} to $S_2$ and $Z$ to obtain
	\small
	\begin{align*}
	R(S)=&\frac{\alpha(S_1)U(S_1)}{U(S)+u_0}+\frac{\alpha(S_2)U(S_2)}{U(S)+u_0}\cdot\left(1-\frac{U(S_1)}{U(S)+u_0}\right)\\
	=&\frac{\alpha(S_1)U(S_1)}{U(S)-U(Z)+u_0}\cdot\left(1-\lambda(Z,S)\right)+\frac{\alpha(Z)U(Z)(U(S_2)+u_0)}{(U(S)+u_0)^2}\\
	&+\frac{(\alpha(S_2)U(S_2)-\alpha(Z)U(Z))(U(S_2)+u_0)}{(U(S)+u_0)^2}\\
	=&\frac{\alpha(S_1)U(S_1)}{U(S)-U(Z)+u_0}\cdot\left(1-\lambda(Z,S)\right)+\frac{\alpha(S_2\setminus Z)(U(S_2)-U(Z))(U(S_2)-U(Z)+u_0)}{(U(S)+u_0)^2}\\
	&+\frac{\alpha(S_2\setminus Z)(U(S_2)-U(Z))U(Z)}{(U(S)+u_0)^2}+\frac{\alpha(Z)U(Z)(U(S_2)+u_0)}{(U(S)+u_0)^2},
	\end{align*}
	\normalsize
	
	\noindent
	where we multiply and divide the first term by $(U(S)-U(Z)+u_0)$ and
	use the definition of $\lambda(Z,S)$ and then add and subtract
	$\frac{\alpha(Z)U(Z)(U(S_2)+u_0)}{(U(S)+u_0)^2}$. The last step uses
	Equation \eqref{identity} and adds and subtracts
	$\frac{\alpha(S_2\setminus Z)(U(S_2)-U(Z))U(Z)}{(U(S)+u_0)^2}$. The
	goal of this manipulation is to form $R(S \setminus
	Z)\left(1-\lambda(Z,S)\right)$. We then obtain
	
	\small
	\begin{align*}
	=&\frac{\alpha(S_1)U(S_1)}{U(S)-U(Z)+u_0}\cdot\left(1-\lambda(Z,S)\right)+\frac{\alpha(S_2\setminus Z)(U(S_2)-U(Z))(U(S_2)-U(Z)+u_0)}{(U(S)+u_0)^2}\\
	&+\frac{\alpha(S_2\setminus Z)(U(S_2)-U(Z))U(Z)}{(U(S)+u_0)^2}+\frac{\alpha(Z)U(Z)(U(S_2)+u_0)}{(U(S)+u_0)^2}\\
	=&\frac{\alpha(S_1)U(S_1)}{U(S)-U(Z)+u_0}\cdot\left(1-\lambda(Z,S)\right)+\frac{\alpha(S_2\setminus Z)(U(S_2)-U(Z))(U(S_2)-U(Z)+u_0)}{(U(S)-U(Z)+u_0)^2}\left(1-\lambda(Z,S)\right)^2\\
	&+\frac{\alpha(S_2\setminus Z)(U(S_2)-U(Z))U(Z)}{(U(S)+u_0)^2}+\frac{\alpha(Z)U(Z)(U(S_2)+u_0)}{(U(S)+u_0)^2}\\
	=&\underbrace{\left[\frac{\alpha(S_1)U(S_1)}{U(S)-U(Z)+u_0}+\frac{\alpha(S_2\setminus Z)(U(S_2)-U(Z))(U(S_2)-U(Z)+u_0)}{(U(S)-U(Z)+u_0)^2}\right]}_{R(S_1\cup S_2\setminus Z)}\cdot\left(1-\lambda(Z,S)\right)\nonumber\\
	&+\frac{\alpha(Z)U(Z)(U(S_2)+u_0)}{(U(S)+u_0)^2}+\frac{\alpha(S_2\setminus Z)(U(S_2)-U(Z))U(Z)}{(U(S)+u_0)^2}\nonumber\\
	&-\frac{\lambda(Z,S)\left(1-\lambda(Z,S)\right)\alpha(S_2\setminus Z)(U(S_2)-U(Z))(U(S_2)-U(Z)+u_0)}{(U(S)-U(Z)+u_0)^2},
	\end{align*}
	\normalsize where the second step multiplies and divides the second
	term by $(U(S)-U(Z)+u_0)$ and uses the definition of
	$\lambda(Z,S)$. The third step splits the second term by expressing
	$(1-\lambda(Z,S))^2$ as $(1-\lambda(Z,S)) -
	\lambda(Z,S)(1-\lambda(Z,S))$ in order to identify the term $R(S_1\cup
	S_2\setminus Z)$. Now, putting together the remaining terms, we
	obtain:
	
	\small
	\begin{align*}
	R(S) =& R(S_1\cup S_2\setminus Z)\left(1-\lambda(Z,S)\right)\\
	+&\lambda(Z,S)\Biggl[\frac{\alpha(Z)(U(S_2)+u_0)}{U(S)+u_0}+\frac{\alpha(S_2\setminus Z)(U(S_2)-U(Z))}{U(S)+u_0}\\
	&-\frac{\left(1-\lambda(Z,S)\right)\alpha(S_2\setminus Z)(U(S_2)-U(Z))(U(S_2)-U(Z)+u_0)}{(U(S)-U(Z)+u_0)^2}\Biggr]\\
	=&  R(S_1\cup S_2\setminus Z)\left(1-\lambda(Z,S)\right)\\
	&+\lambda(Z,S)\cdot\left[\frac{\alpha(Z)(U(S_2)+u_0)}{U(S)+u_0}+\frac{\alpha(S_2\setminus Z)(U(S_2)-U(Z))}{U(S)+u_0}\left[1-\left(1-\frac{U(S_1)}{U(S)-U(Z)+u_0}\right)\right]\right]\\
	=& R(S \setminus Z)\left(1-\lambda(Z,S)\right)+\left[\frac{\alpha(Z)(U(S_2)+u_0)}{U(S)+u_0}+\frac{\alpha(S_2\setminus Z)(U(S_2)-U(Z))}{U(S)-U(Z)+u_0}\cdot\frac{U(S_1)}{U(S)+u_0}\right]\cdot\lambda(Z,S).
	\end{align*}
	\normalsize
	
	\noindent			
	where the second line uses the definition of $\lambda(Z,S)$ to factorize the expression and the last step just simplifies the resulting expressions.
\end{proof}

	Now we can continue with the proof of Proposition \ref{global_bound}. Let $Z\subseteq S_i^*$ for some $i=1,2$. Consider first the case in
	which the optimal solution $S^*$ contains only products from level
	$i$ $(i \in \{1,2\})$. Then,
	\begin{align*}
	R^*&=\frac{\sum_{y \in S^*}r(y)u(y)}{\sum_{y\in S^*}u(y)+u_0}\\
	&=\frac{\sum_{y \in S^*\setminus Z}r(y)u(y)}{\sum_{y\in S^*\setminus Z}u(y)+u_0}\cdot\frac{\sum_{y\in S^*}u(y) -U(Z)+u_0}{\sum_{y\in S^*}u(y)+u_0} + \frac{\alpha(Z)U(Z)}{\sum_{y\in S^*}u(y)+u_0}\\
	&=\underbrace{\frac{\sum_{y \in S^*\setminus Z}r(y)u(y)}{\sum_{y\in S^*\setminus Z}u(y)+u_0}}_{R(S^*\setminus Z)}\cdot\left(1-\lambda(Z,S^*)\right) + \alpha(Z)\lambda(Z,S^*)\\
	&=R(S^*\setminus Z)\cdot\left(1-\lambda(Z,S^*)\right) + \alpha(Z)\lambda(Z,S^*).
	\end{align*}
	
	\noindent	
	The optimal solution is a convex combination of $R(S^*\setminus Z)$ and $\alpha(Z)$. By optimality of $R^*$, $R(S^*\setminus Z)\leq R^*$ and hence $\alpha(Z)\geq R^*$.
	
	Consider the case in which the solution is non-empty in both levels,
	and suppose that $\alpha(Z)<R^*$. We now show that this is not
	possible. The proof considers two independent cases, depending on the
	level that contains $Z$.
	
	If $Z\subseteq S^*_1$, by Lemma
	\ref{convex_lemma}, the revenue of $S^*$ can be expressed as
	\small
	\begin{equation}
	R(S^*)= R(S^*\setminus Z)\cdot(1-\lambda(Z,S^*)) +\underbrace{\left[\alpha(Z)-\frac{\alpha(S^*_2)U(S_2)(U(S^*_2)+u_0)(1-\lambda(Z,S^*)}{(U(S^*)-U(Z)+u_0)^2}\right]}_{\Gamma_Z}\cdot\lambda(Z,S^*).
	\end{equation}
	\normalsize
	
	\noindent
	$R^*$ is a convex combination of $R(S^*\setminus Z)$ and $\Gamma_Z$. We show that $\Gamma_Z < R^*$.
	\small
	\[
	\begin{aligned}
	\Gamma_Z&= \alpha(Z)-\frac{\alpha(S_2^*)U(S_2^*)(U(S_2^*)+u_0)\left(1-\lambda(Z,S^*)\right)}{(U(S^*)-U(Z)+u_0)^2} \\
	&=\alpha(Z)-\frac{\alpha(S_2^*)U(S_2^*)(U(S_2^*)+u_0)}{(U(S^*)-U(Z)+u_0)(U(S^*)+u_0)}&& \text{/using definition of $\lambda(Z,S^*)$}\\
	&=\alpha(Z)-\frac{\alpha(S_2^*)U(S_2^*)(1-\lambda(S_1^*,S^*))}{(U(S^*)-U(Z)+u_0)}&& \text{/by definition of $\lambda(S_1^*,S^*)$ }\\
	&\leq \alpha(Z)-\frac{R^*U(S_2^*)}{U(S^*)-U(Z)+u_0}&& \text{/Using proposition \ref{R2_bound}}\\
	&< R^*\left(1-\frac{U(S_2^*)}{U(S^*)-U(Z)+u_0}\right)&& \text{/using the assumption $\alpha(Z)<R^*$}\\
	&<R^*.
	\end{aligned}
	\]
	\normalsize
	
	\noindent
	Since $R(S^*/Z) \leq R^*$, we have that $R^*<R^*$ and hence it must be the case that $\alpha(Z)\geq R^*$.

	\noindent Now, if $Z\subseteq S^*_2$, by Lemma \ref{convex_lemma}, the revenue of
	$S^*$ can be expressed as
	\small
	\begin{align}
	R(S^*)&=R(S^*\setminus Z)\left(1-\lambda(Z,S^*)\right)\nonumber\\
	&+\underbrace{\left[\frac{\alpha(Z)(U(S_2^*)+u_0)}{U(S^*)+u_0}+\frac{\alpha(S_2^*\setminus Z)(U(S_2^*)-U(Z))}{U(S^*)-U(Z)+u_0}\cdot\frac{U(S_1^*)}{U(S^*)+u_0}\right]}_{\Gamma_Z}\cdot\lambda(Z,S^*).
	\end{align}
	\normalsize
	
	\noindent
	$R^*$ is thus a convex combination of $R(S^*\setminus Z)$ and $\Gamma_Z$. Again, we show that $\Gamma_Z < R^*$:
	\small
	\[
	\begin{aligned}
	\Gamma_Z&= \frac{\alpha(Z)(U(S_2^*)+u_0)}{U(S^*)+u_0}+\frac{\alpha(S_2^*\setminus Z)(U(S_2^*)-U(Z))}{U(S^*)-U(Z)+u_0}\cdot\frac{U(S_1^*)}{U(S^*)+u_0}\\
	&= \alpha(Z)(1-\lambda(S_1^*,S^*))+\frac{\alpha(S_2^*\setminus Z)(U(S_2^*)-U(Z))}{U(S^*)-U(Z)+u_0}\cdot \lambda(S_1^*,S^*) && \text{/by definition of $\lambda(S_1^*,S^*)$}\\
	&< \alpha(Z)(1-\lambda(S_1^*,S^*))+\underbrace{\frac{\alpha(S_2^*\setminus Z)(U(S_2^*)-U(Z))}{U(S_2^*)-U(Z)+u_0}}_{R(S_2^*\setminus Z)}\cdot \lambda(S_1^*,S^*) && \text{/replacing $U(S^*)$ by $U(S_2^*)$ }\\
	&< R^*(1-\lambda(S_1^*,S^*)) + R(S_2^*\setminus Z)\cdot \lambda(S_1^*,S^*) && \text{/using that $\alpha(Z)<R^*$}\\
	&< R^*(1-\lambda(S_1^*,S^*)) +  R(S_2^*\setminus Z)\lambda(S_1^*,S^*) && \text{/Using the optimality of $R^*$}\\
	&< R^*.
	\end{aligned}
	\]
	\normalsize
	\noindent
	Hence, it must be the case that $\alpha(Z)\geq R^*$, completing the proof.

\endproofatend

\noindent
The converse of Proposition \ref{global_bound} does not hold: Example
\ref{ex:notallincluded} presents an
instance where the optimal solution does not contain all the products
with a revenue higher than $R^*$.

\begin{exmp}
	\label{ex:notallincluded}
	We show that some products with revenue greater or equal than $R^*$
	may not be included in an optimal assortment.  Let
	$X_1=\left\{x_{11}\right\}$, $X_2=\left\{x_{21}\right\}$, and
	$X=X_1\uplus X_2$. Let the revenues be $r(x_{11})=r(x_{21})=1$ and
	the utilities be $u(x_{11})=10, u(x_{21})=1$, and $u_0=1$. Consider
	the possible assortments and their expected revenues: \\
	
	\begin{tabular}{ll}
		$R(\left\{x_{11}\right\})$& $= \frac{u(x_{11})r(x_{11})}{u(x_{11})+u_0}=\frac{10\cdot 1}{10+1} =0.9\overline{09}$ \\
		$R(\left\{x_{21}\right\})$& $= \frac{u(x_{21})r(x_{21})}{u(x_{21})+u_0}=\frac{1\cdot 1}{1+1} =0.5$ \\
		$R(\left\{x_{11},x_{21}\right\})$&= $\frac{u(x_{11})r(x_{11})}{u(x_{11})+u(x_{21})+u_0}+\left(1-\frac{u(x_{11})}{u(x_{11})+u(x_{21})+u_0}\right)\cdot\frac{u(x_{21})r(x_{21})}{u(x_{11})+u(x_{21})+u_0}$\\
		&$=\frac{10\cdot 1}{10+1+1}+\left(1-\frac{10}{10+1+1}\right)\cdot\frac{1\cdot 1}{10+1+1}$\\
		&$=\frac{10}{12}+\left(1-\frac{10}{12}\right)\cdot\frac{1}{10+1}$\\
		&$= 0.847\overline{2}$
	\end{tabular}

	\noindent
	The optimal assortment is $S^*=\left\{x_{11}\right\}$. However, we have
	that $r(x_{21})=1>R^*$, but $x_{21}$ is not part of the optimal assortment.
\end{exmp}

The following corollary, whose proof
is also in Appendix \ref{proofs}, is a direct consequence of Proposition
\ref{global_bound}.
	
\begin{corollary}
\label{subset_bound}
For any non-empty subset $S_0\subseteq S^*$, where $S^*$ is an optimal solution, we have $\alpha(S_0)\geq R^*$.
\end{corollary}
\proofatend By Proposition \ref{global_bound}, we have
$\alpha(\left\{x\right\})=r(x)\geq R^*$ for all $x\in S^*$. For each set $S_0\subseteq S^*$, we
have:
\begin{align*}
\alpha(S_0)=&\frac{\sum_{x\in S_0} u(x)r(x)}{\sum_{x\in S_0}u(x)}\\
  	\geq& R^*\frac{\sum_{x\in S_0} u(x)}{\sum_{x\in S_0}u(x)}\quad \text{/ by Proposition \ref{global_bound}}\\
	=&R^*.
\end{align*}
 \endproofatend
	
\noindent
The corollary above implies that $\alpha(\left\{x\right\})=r(x)\geq
R^*$ for all $x\in S^*$. Thus, every product in an optimal assortment
has a revenue greater than or equal to $R^*$.

In the following example we show that the well known revenue-ordered assortment strategy for the assortment problem does not always lead to optimality.

\begin{exmp}[Revenue-Ordered assortments are not optimal]
	\label{ex:RO_not_optimal} Let
	$X_1=\left\{x_{11}\right\}$, $X_2=\left\{x_{21}\right\}$, and
	$X=X_1\uplus X_2$. Let $r(x_{11})=10$ and $r(x_{21})=12$. Let
	the utilities be $u(x_{11})=10, u(x_{21})=2$, and $u_0=1$. A direct calculation shows that the
	 revenues for all possible assortments under this setting are:
		
  \begin{center}\label{tab:RO_not_optimal}
	\begin{tabular}{@{} cc @{}}
		\toprule
		$S$ &        $R(S)$  \\
		\midrule
		$\{x_{11}\}$  & $9.\overline{09}$    \\	
		$\{x_{21}\}$ &   8 \\
		$\{x_{11},x_{21}\}$  & $ 8.12$    \\	
		\bottomrule
	\end{tabular}
\end{center}
	
	\noindent
	The optimal assortment is $S^*=\left\{x_{11}\right\}$, yielding a revenue of $R^*=9.\overline{09}$. However,  the best revenue ordered assortment is $S'=\{x_{11},x_{21}\}$, obtaining a revenue of $R'= 8.12$, this means that the revenue-ordered assortment strategy provides an approximation ratio of $\frac{R'}{R^*}\approx 89.3\%$ for this particular instance.
\end{exmp}

Up to this point, we understand some properties of optimal assortments, but we don't know the solution structure or an algorithm to calculate it. In the following section we propose a natural extension of the revenue-ordered assortment strategy, the \textit{revenue-ordered by level assortments}, and show that guarantees optimality, and it can be computed in polynomial time.

\section{Optimality of Revenue-Ordered Assortments by level}
\label{sec:optimality}

This section proves that optimal assortments under the SML are
revenue-ordered by level, generalizing the traditional results for the
MNL \citep{talluri2004revenue}. As a corollary, the optimal assortment
problem under the SML is polynomial-time solvable.

\begin{defn}[Revenue-Ordered Assortment by Level] Denote by $N_{ij}$
  the set of all products in level $i$ with a revenue of at least
  $r_{ij}$ ($1 \leq j \leq m_i$) and fix $N_{i0}=\emptyset$ by convention. A
  revenue-ordered assortment by level is a set $S=N_{1j_1} \uplus
  N_{2j_2} \subseteq X$ for some $0 \leq j_1\leq m_1$ and $0\leq j_2 \leq m_2$. We use
  $\mathcal{A}$ to denote the set of revenue-ordered assortments by
  level.
\end{defn}

\noindent
When an assortment $S = S_1 \uplus S_2$ is not revenue-ordered by level, it follows that
\[
\exists k \in \{1,2\}, z \in X_k \setminus S_k, y \in S_k: \quad r(z) \geq r(y).
\]
We say that $S$ has a gap, the gap is at level $k$, and $z$ belongs to the gap. We now define
the concept of {\em first gap}, which is heavily used in the proof.

\begin{defn}[First Gap of an Assortment]
\label{firstgap_def}
Let $S = S_1 \uplus S_2$ be an assortment with a gap and let $k$ be the smallest level with a gap.
Let $\hat{r}=\max_{y\in X_k\setminus S_k}r(y)$ be the maximum revenue of a product in level $k$ not contained in $S_k$. The \emph{first gap} of $S$ is a set of products $G \subseteq X_k\setminus S$ defined as follows:
\begin{itemize}
\item If $\max_{x\in S_k}r(x)<\hat{r}$, then the gap $G$ consists of
  all products with higher revenues than the products in assortment
  $S$, i.e.,
\[
G=\{y\in X_k\setminus S_k\mid r(y)\geq\max_{x\in S_k}r(x)\}.
\]

\item Otherwise, when $\max_{x\in S_k}r(x)\geq\hat{r}$, define the following quantities:
\begin{equation}\label{piv_rev}
r_M=\min_{\substack{x\in S_k\\r(x)\geq \hat{r}}}r(x)\quad \text{ and }\quad r_m=\max_{\substack{x\in S_k\\r(x)\leq \hat{r}}}r(x).
\end{equation}
The set $G$ contains products with revenues in $\lbrack r_m,r_M\rbrack$, i.e.,
\[
G=\{y\in X_k\setminus S_k\mid r_m\leq r(y)\leq r_M\}.
\]
\end{itemize}
\end{defn}

\noindent
We are now in a position to state the main theorem of this paper.

\begin{thm}
\label{main_thm}
Under the SML, any optimal assortment is revenue-ordered by level.
\end{thm}

The proof is fairly technical, and is relegated to Appendix \ref{proofs}.
The intuition behind it is the following: Assume that $S$ is an optimal solution with at least one gap as in Definition \ref{firstgap_def}. Let
$G$ be the first gap of $S$, and that $G$ occurs at level $k$.  Define
$S_k=H\cup T$ with $H,T \subseteq X_k$ and
\[
H=\{x \in S_k \, \mid \, r(x)\geq \max_{g\in G}r(g)\}
\]
and
\[
T=\{x \in S_k \, \mid \, r(x)\leq \min_{g\in G}r(g)\}.
\]
We call the set $H$ as the {\em head} and the set
$T$ is called the {\em tail}.  We prove that is always possible to
select an assortment that is revenue-ordered by level and has
revenue greater than $R(S)$ (contradicting optimality).  The proof shows that such an
assortment can be obtained either by including the gap $G$ in $S$ or
by eliminating $T$ from $S$. Figure \ref{fig:stats} illustrates
these concepts visually.

\begin{figure}[!ht]
\centering
\begin{tikzpicture}
\begin{axis}[
xbar stacked,
legend style={
legend columns=4,
				at={(xticklabel cs:0.5)},
				anchor=north,
				draw=none
			},
			ytick=data,
			axis y line*=none,
			axis x line*=bottom,
			tick label style={font=\footnotesize},
			legend style={font=\footnotesize},
			label style={font=\footnotesize},
			xtick style={draw=none},
			ytick style={draw=none},
			title={Graphical representation of a gap on a fixed level},
			xtick={10,590},
			xticklabels={$+$,$-$},
			width=.9\textwidth,
			bar width=6mm,
			xlabel={Revenue},
			yticklabels={$H$, $H\cup G\cup T$ , , $H\cup T$},
			xmin=0,
			xmax=600,
			area legend,
			y=8mm,
			enlarge y limits={abs=0.625},
			]
			\addplot[BORD,fill=ALLIN] coordinates
			{(150,0) (200,1) (0,2) (150,3)};
			\addplot[BORD,fill=EMPTY] coordinates
			{(0,0) (0,1) (0,2) (50,3) };
			\addplot[BORD,fill=POSSIBLY] coordinates
			{(0,0) (150,1) (0,2) (150,3) };
			
			\legend{All products on the assortment, No products selected, Possibly some products missing}
			\coordinate (original) at (500,3);
			\coordinate (explanation) at (200,2);
			\coordinate (drop) at (300,0);
			\coordinate (fill) at (475,1);
			\coordinate (start) at (0,-0.625);
			\coordinate (end) at (600,-0.625);
			\coordinate (end_rev) at (635,-0.625);
			
			\end{axis}
			\node at (drop) {Candidate 2, $T$ is removed};
			\node at (fill) {Candidate 1, $G$ is added};
			\node at (original) {Original assortment};
			\node at (end_rev) {Revenue};
			\draw[->] (start) -- (end);

			\draw [decorate,decoration={brace,amplitude=5pt},xshift=0pt,yshift=6pt] (0,3) -- (3.3,3) node [black,midway,yshift=10pt] {\footnotesize $H$};
			\draw [decorate,decoration={brace,amplitude=5pt},xshift=0.05pt,yshift=6pt] (3.3,3) -- (4.4,3) node [black,midway,yshift=10pt] {\footnotesize $G$};
			\draw [decorate,decoration={brace,amplitude=5pt},xshift=0.2pt,yshift=6pt] (4.4,3) -- (7.75,3)node [black,midway,yshift=10pt] {\footnotesize $T$};
			\end{tikzpicture}
			\caption{Representation of a level containing a gap $G$ at the top, and the two proposed candidates fixing the gap by either adding $G$, or removing $T$. }
			\label{fig:stats}
		\end{figure}

\proofatend
  Assume that $S$ is an optimal solution with at least one gap,
  $G$ is the first gap of $S$, and $G$ occurs at level $k$.  Define
  $S_k=H\cup T$ with $H,T \subseteq X_k$ and
  \[
      H=\{x \in S_k \, \mid \, r(x)\geq \max_{g\in G}r(g)\}
  \]
  and
  \[
      T=\{x \in S_k \, \mid \, r(x)\leq \min_{g\in G}r(g)\}.
  \]
  In the following, the set $H$ is called the {\em head} and the set
  $T$ is called the {\em tail}.  We prove that is always possible to
  select an assortment that is revenue-ordered by level and has
  revenue greater than $R(S)$.  The proof shows that such an
  assortment can be obtained either by including the gap $G$ in $S$ or
  by eliminating $T$ from $S$. Figure \ref{fig:stats} illustrates
  these concepts visually. The proof is by case analysis on the level
  of $G$.

Consider first the case where $G$ is in the first
level. We can define $S=S_1 \uplus S_2$, with $S_1=H
\cup T$ as defined above and $S_2\subseteq X_2$. The
revenue for $S$ is

\begin{align*}\label{rev_gapopt1}
R(S_1\cup S_2)&=\frac{\alpha(S_1)U(S_1)}{U(S)+u_0}+ \left(1-\frac{U(S_1)}{(U(S)+u_0)}\right)\cdot \frac{\alpha(S_2)U(
S_2)}{U(S)+u_0}\\
&=\frac{\alpha(H)U(H)}{U(S)+u_0}+\frac{\alpha(T)U(T)}{U(S)+u_0} + \frac{U(S_2)+u_0}{(U(S)+u_0)^2}\cdot \alpha(S_2)U(S_2)
\end{align*}
			
\noindent
where we used Proposition \ref{decomposition} on $S_1$ for deriving
the second equality. We show that assortment $H\cup S_2$ or assortment
$H\cup G\cup T\cup S_2$ provides a revenue greater than $R(S = H\cup
T\cup S_2)$, contradicting our optimality assumption for $S$. The proof
characterizes the differences between the revenues of $S$ and the two
considered assortments, adds those two differences, and shows that
this value is strictly less than zero, implying that at least one of
the differences is strictly negative and hence that one of these
assortments has a revenue larger than $R(S)$. $R(H
\cup S_2)$ can be expressed as

\small
\begin{equation}
\label{rev_s1Ut}
\frac{\alpha(H)U(H)}{U(H)+U(S_2)+u_0}+ \frac{U(S_2)+u_0}{(U(H)+U(S_2)+u_0)^2}\cdot \alpha(S_2)U(S_2).
\end{equation}
\normalsize Let $\theta=U(H)+U(T)+U(S_2)+u_0$ (or, equivalently,
$\theta=U(S)+u_0$). The difference $R(H\cup T \cup
S_2) - R(H\cup S_2)$ is
\small
\begin{align}\label{rev_diff1}
& \frac{U(T)}{\theta(\theta-U(T))}\left[-\alpha(H)U(H)+\alpha(T)(\theta-U(T))-\frac{\alpha(S_2)U(S_2)(U(S_2)+u_0)(2(\theta-U(T)) + U(T))}{\theta(\theta - U(T))}\right].
\end{align}		
\normalsize

\noindent
$R(H\cup G\cup T\cup S_2)$ can be expressed as
\small
\begin{align*}
& \frac{1}{U(S)+U(G)+u_0}\cdot\left[\alpha(H)U(H)+\alpha(G)U(G)+\alpha(T)U(T)\right] +\frac{U(S_2)+u_0}{(U(S)+U(G)+u_0)^2}\cdot \alpha(S_2)U(S_2).
\end{align*}
\normalsize

\noindent				
The difference $R(H\cup T\cup S_2) - R(H\cup G\cup T\cup S_2)$ is given by
\small
\begin{align}\label{rev_diff2}
&\frac{U(G)}{\theta(\theta+U(G))}\cdot\left[\alpha(H)U(H)+\alpha(T)U(T)-\alpha(G)\theta+\frac{\alpha(S_2)U(S_2)(U(S_2)+u_0)(2\theta + U(G))}{\theta(\theta+U(G))}\right].
\end{align}	
\normalsize	

\noindent
By optimality of $S$, these two differences must be positive. However,
their sum, dropping the positive multiplying term on each difference,
which must also be positive, is given by \small
\begin{align*}\label{sum_revs_proof}
(\ref{rev_diff1})+(\ref{rev_diff2})&=\alpha(T)\theta-\alpha(G)\theta+\frac{\alpha(S_2)U(S_2)(U(S_2)+u_0)}{\theta}\cdot\left[\frac{2\theta+U(G)}{\theta+U(G)}-\frac{2(\theta-U(T)) +U(T)}{(\theta-U(T))}\right]\\
&=(\alpha(T)-\alpha(G))\theta +\frac{\alpha(S_2)U(S_2)(U(S_2)+u_0)}{\theta}\cdot\left[1+\frac{\theta}{\theta+U(G)}-2-\frac{U(T)}{(\theta-U(T))}\right]\\
&=\underbrace{(\alpha(T)-\alpha(G))\theta}_{\leq 0, \text{ by Proposition \ref{bound_alpha} and }\theta\geq 0}+\frac{\alpha(S_2)U(S_2)(U(S_2)+u_0)}{\theta}\cdot\left[\underbrace{\left(\frac{\theta}{\theta+U(G)}-1\right)}_{<0} \underbrace{-\frac{U(T)}{(\theta-U(T))}}_{<0}\right]< 0
\end{align*}	
\normalsize
which contradicts the optimality of $S$.
			
Consider now the case where the gap is in the second level. Using the definition of the head and the tail discussed above, $S$ can be written as
$
S = S_1 \uplus H \cup T.
$
The revenue $R(S)$ is given by
\begin{equation}\label{rev_gapopt2}
R(S_1\cup H \cup T)=\frac{\alpha(S_1)U(S_1)}{\theta} +\frac{\alpha(H)U(H)}{\theta}+\frac{\alpha(T)U(T)}{\theta} -\frac{U(S_1)\alpha(H)U(H)}{\theta^2}-\frac{U(S_1)\alpha(T)U(T)}{\theta^2}
\end{equation}

\noindent
and the proof follows the same strategy as for the case of the first
level. The revenue $R(S_1\cup H)$ is given by
\begin{equation}\label{rev_sUt1}
R(S_1\cup H)=\frac{\alpha(S_1)U(S_1)}{\theta-U(T) } +\frac{\alpha(H)U(H)}{\theta-U(T) } -\frac{U(S_1)\alpha(H)U(H)}{(\theta-U(T) )^2}
\end{equation}

\noindent
and the difference $R(S_1\cup H \cup T) - R(S_1\cup H)$ by
\small
\begin{align}\label{rev_diff21}
& \frac{U(T) }{\theta(\theta-U(T) )}\cdot\Biggl[-\alpha(S_1)U(S_1) -\alpha(H)U(H)+\alpha(T)(\theta-U(T))\nonumber\\
&+\frac{\alpha(H)U(H)U(S_1)(2\theta-U(T))}{\theta(\theta - U(T))}-\frac{U(S_1)\alpha(T)(\theta-U(T))}{\theta}\Biggr]
\end{align}		
\normalsize
The revenue $R(S_1\cup H \cup G\cup T)$ is given by
\small
\begin{align}\label{rev_sut1ugut2}
& \frac{\alpha(S_1)U(S_1)}{\theta+U(G)} +\frac{\alpha(H)U(H)}{\theta+U(G)}+\frac{\alpha(G)U(G)}{\theta+U(G)}+\frac{\alpha(T)U(T)}{\theta+U(G)} -\frac{U(S_1)}{(\theta+U(G))^2}\cdot\left[\alpha(H)U(H)+\alpha(G)U(G)+\alpha(T)U(T)\right]
\end{align}
\normalsize
\noindent				
and the difference $R(S_1\cup H \cup T) - R(S_1\cup H \cup G\cup T)$ by
\small
\begin{align}\label{rev_diff22}
&\frac{U(G)}{\theta(\theta+U(G))}\cdot\Biggl[\alpha(S_1)U(S_1)+\alpha(H)U(H)-\alpha(G)\theta +\alpha(T)U(T)\nonumber\\
&-\frac{\alpha(H)U(H)U(S_1)(2\theta+U(G))}{\theta(\theta+U(G))}-\frac{\alpha(T)U(T)U(S_1)(2\theta+U(G))}{\theta(\theta+U(G))}+\frac{U(S_1)\alpha(G)\theta}{\theta+U(G)}\Biggr]
\end{align}	
\normalsize

\noindent
Adding \eqref{rev_diff21} and \eqref{rev_diff22} and dropping the positive multiplying terms on each difference gives
\small
\begin{align}\label{Second_level_ineq}
&\theta\left(\alpha(T)-\alpha(G)\right)+U(S_1)\left(\alpha(G)-\alpha(T)\right)+U(S_1)\left[\frac{\alpha(T)U(T)}{\theta}-\frac{\alpha(G)U(G)}{\theta+U(G)}\right]\nonumber\\
			& +\frac{U(S_1)}{\theta}\left[\alpha(H)U(H)\left(1+\frac{\theta}{\theta-U(T)}\right)-\alpha(H)U(H)\left(1+\frac{\theta}{\theta+U(G)}\right)-\alpha(T)U(T)\left(1+\frac{\theta}{\theta+U(G)}\right)\right]\nonumber\\
			&= \left(\theta-U(S_1)\right)\left(\alpha(T)-\alpha(G)\right)+\frac{U(S_1)\alpha(H)U(H)}{\theta-U(T)}-\frac{U(S_1)}{\theta+U(G)}\cdot\left[\alpha(H)U(H)+\alpha(G)U(G)+\alpha(T)U(T)\right] \nonumber\\
			&= \underbrace{\left(\theta-U(S_1)\right)\left(\alpha(T)-\alpha(G)\right)}_{\leq 0, \text{ by Proposition \ref{bound_alpha} and }\theta\geq U(S_1)}
			+\underbrace{\frac{U(S_1)\alpha(H)U(H)(U(G)+U(T))}{(\theta-U(T))(\theta+U(G))}-\frac{U(S_1)}{\theta+U(G)}\cdot\left[\alpha(G)U(G)+\alpha(T)U(T)\right]}_{\Gamma}
			\end{align}	
			\normalsize
\noindent
$\Gamma$ cannot be greater or equal than zero, since otherwise
			
\begin{align}\label{gamma_gap2}
&\frac{U(S_1)\alpha(H)U(H)(U(G)+U(T))}{(\theta-U(T))(\theta+U(G))}-\frac{U(S_1)}{\theta+U(G)}\cdot\left[\alpha(G)U(G)+\alpha(T)U(T)\right]\geq 0 \nonumber\\
&\frac{U(S_1)}{(\theta+U(G))}\cdot \left[\frac{\alpha(H)U(H)(U(G)+U(T))}{(\theta-U(T))}-\left(\alpha(G)U(G)+\alpha(T)U(T)\right)\right]\geq 0.
\end{align}
	
\noindent	
The factor on the left is always positive, so Inequality
\eqref{gamma_gap2} implies that the term between brackets is greater
than zero. We now show that this contradicts the optimality of $S$. We
do this by manipulating Inequality \eqref{gamma_gap2} and showing that,
if this inequality holds, then $R(H)>R(S)$.

\small
\begin{align}\label{proof_end2}
			&\frac{\alpha(H)U(H)(U(G)+U(T))}{(\theta-U(T))}-\left(\alpha(G)U(G)+\alpha(T)U(T)\right)\geq 0\nonumber\\
&\frac{\alpha(H)U(H)}{(\theta-U(T))}\geq\frac{\alpha(G)U(G)+\alpha(T)U(T)}{(U(G)+U(T))}\nonumber\\
&\frac{\alpha(H)U(H)}{(U(S_1)+U(H)+u_0)}\geq\frac{\alpha(G)U(G)+\alpha(T)U(T)}{(U(G)+U(T))}\nonumber\\
&\underbrace{\frac{\alpha(H)U(H)}{(U(H)+u_0)}}_{R(H)}\cdot\left[1-\frac{U(S_1)}{U(S_1)+U(H)+u_0}\right]\geq\frac{\alpha(G)U(G)+\alpha(T)U(T)}{(U(G)+U(T))}\nonumber\\
& R(H)\geq\underbrace{R(H)\cdot\frac{U(S_1)}{U(S_1)+U(H)+u_0}}_{>0}+\frac{\alpha(G)U(G)+\alpha(T)U(T)}{(U(G)+U(T))}>R(S) \cdot \frac{U(G)+U(T)}{U(G)+U(T)}>R(S).
\end{align}

\normalsize
\noindent
Inequality \eqref{proof_end2} follows from Proposition
\ref{global_bound} applied to $T\subset S_2$, which implies
$\alpha(T)\geq R(S)$, and from Proposition \ref{bound_alpha}, which implies
$\alpha(G)\geq \alpha(T)$ and hence $\alpha(G)\geq R(S)$.

\endproofatend

\noindent	
The following corollary follows directly from the fact that there are
at most $\mathcal{O}(|X|^2)$ revenue-ordered assortments by level and
the fact that the revenue obtained from a given assortment can be computed in
polynomial time.
	
\begin{corollary}
\label{time_complexity}
The assortment problem under the sequential multinomial logit is polynomial-time solvable.
\end{corollary}

\section{Numerical Experiments}\label{sec:numerical}

In this section, we analyse numerically the performance of revenue-ordered assortments (RO) against our proposed strategy (ROL) by varying the number of products, the distribution of revenues and utilities in each level, and the utility of the outside option. In our experiments with up to $100$ products, we found that the optimality gap can be as large as $26.319\%$.

Each family or class of instances we tested is defined by three numbers: the number of products in the first and second level $(n_1,n_2)$, and the utility of the outside option $u_0$. In total, we tested 20 classes or family of instances, each containing $100$ instances. In each specific instance, revenues and product utilities are drawn from an uniform distribution between $0$ and $10$. We ran both strategies (RO and ROL) and we report the average and the worst optimality gap for the RO strategy, and the time taken for both strategies. These numerical experiments were conducted in Python 3.6, at a computer with 4 processors (each with 3.6 GHz CPU) and 16 GB of RAM. The computing time is reported in seconds is the average among the 100 instances in each class (or family).

\begin{table}[!htbp]
	\centering
	\begin{tabular}{cc@{\qquad}ccc@{\qquad}cc}
	\toprule
	\multirow{2}{*}{\raisebox{-\heavyrulewidth}{$(n_1,n_2)$}} &\multirow{2}{*}{\raisebox{-\heavyrulewidth}{$u_0$}} & \multicolumn{3}{c}{RO} & & \multicolumn{1}{c}{ROL} \\
	\cmidrule(r){3-5} \cmidrule(l){6-7}
	& & Avg. Gap $(\%)$ & Worst Gap $(\%)$ & Avg. Time RO (s) & & Avg. Time ROL (s) \\
	\midrule
	(5,5) & 0     & 0     & 0     & 0     &       & 0 \\
	(5,5) & 1     & 1.811 & 18.098 & 0     &       & 0.001 \\
	(5,5) & 2.5   & 3.43  & 17.416 & 0     &       & 0 \\
	(5,5) & 5     & 3.413 & 13.183 & 0     &       & 0 \\
	(5,5) & 10    & 1.923 & 12.406 & 0     &       & 0.001 \\
	(10,10) & 0     & 0     & 0     & 0     &       & 0.004 \\
	(10,10) & 1     & 3.23  & 20.669 & 0.001 &       & 0.004 \\
	(10,10) & 2.5   & 5.613 & 20.359 & 0.001 &       & 0.004 \\
	(10,10) & 5     & 5.975 & 15.521 & 0.001 &       & 0.004 \\
	(10,10) & 10    & 5.347 & 15.694 & 0.001 &       & 0.004 \\
	(20,20) & 0     & 0     & 0     & 0.003 &       & 0.04 \\
	(20,20) & 1     & 4.331 & 19.427 & 0.003 &       & 0.04 \\
	(20,20) & 2.5   & 8.523 & 21.873 & 0.003 &       & 0.039 \\
	(20,20) & 5     & 9.776 & 19.719 & 0.003 &       & 0.038 \\
	(20,20) & 10    & 9.682 & 18.771 & 0.003 &       & 0.039 \\
	(50,50) & 0     & 0     & 0     & 0.016 &       & 0.564 \\
	(50,50) & 1     & 6.315 & 26.319 & 0.016 &       & 0.549 \\
	(50,50) & 2.5   & 11.662 & 24.117 & 0.016 &       & 0.55 \\
	(50,50) & 5     & 14.94 & 24.445 & 0.016 &       & 0.551 \\
	(50,50) & 10    & 15.543 & 22.232 & 0.016 &       & 0.547 \\
	\bottomrule
\end{tabular}%
	\caption{Numerical experiments comparing the revenue ordered assortment strategy (RO) and the revenue-ordered assortments by level (ROL, which is optimal). For each class of instances, we display the average optimality gap and the worst-case gap, as well as the computing time.}
	\label{tab:uniform}%
\end{table}

Based on the results on Table \ref{tab:uniform} we can observe the following:

\begin{enumerate}
	\item As expected, although ROL takes more time than RO, it takes a small amount of time to solve the instances.
	\item When the utility of the outside option is $u_0=0$, the average and worst gap are identically zero. This is because in those cases, the optimal solution is simply selecting the highest revenue product and therefore both strategies coincide.
	\item The average gap is generally increasing as the outside option utility increases. With a high outside option, we typically expect to select more products to counterbalance the effect of the no choice alternative. This can amplify the difference between ROL and RO as the likelihood that the optimal solution of the revenue-ordered by level is indeed a revenue ordered assortment decreases.
\end{enumerate}

\section{Conclusion and Future Work}\label{sec:conclusion}
This paper studied the assortment optimization problem under the
\emph{Sequential Multinomial Logit} (SML), a discrete choice model
that generalizes the multinomial logit (MNL). Under the SML model,
products are partitioned into two levels. When a consumer is presented
with such an assortment, she first considers products in the first
level and, if none of them is appropriate, products in the second
level. The SML is a special case of the Perception Adjusted Luce
Model (PALM) recently proposed by \citet{echenique2018}. It can explain many
behavioural phenomena such as the attraction, compromise, and
similarity effects which cannot be explained by the MNL model or any
discrete choice model based on random utility.

The paper showed that the seminal concept of revenue-ordered
assortments can be generalized to the SML. In particular, the paper
proved that all optimal assortments under the SML are revenue-ordered
by level, a natural generalization of revenue-ordered assortments. As
a corollary, assortment optimization under the SML is polynomial-time
solvable.  This result is particularly interesting given that the SML
does not satisfy the
regularity condition.

The main open issue regarding this research is to generalize the results
to the PALM, which has an arbitrary number of levels. Note that one can easily extend the algorithm of revenue-ordered by level to the PALM model and it would take at most $\mathcal{O}(|X|^k)$ time where $k$ is the number of levels \footnote{Note that $\mathcal{O}(|X|^k)$ is the number assortments that are revenue-ordered by level.}. We executed our algorithm over a series of PALM instances by varying the number of levels and the revenue-ordered assortment algorithm always returned the optimal solution. Our conjecture is that the optimality result of revenue ordered assortments by level holds for the general PALM, but the problem remains open. We note that our proof technique cannot be directly applied to answer this question as some of the bounds developed in Section \ref{sec:properties-2-level} do not hold for $k>2$.

  A second interesting research avenue is to consider a new discrete choice model that allows
   decision makers to change the order in which the
levels are presented to consumers. In the SML, the level ordering
is intrinsic to products, but one may consider settings in
which decision makers can choose, not only what to show, but also the
priority associated with each of the displayed products. Another research direction is to study the assortment optimziation problem under the SML with cardinality or space constraints. Finally, it
is important to develop efficient procedures to estimate the parameters of the SML model based on historical data (e.g., \citet{van2017technical}).

	\newpage
\bibliographystyle{../../../Bib/aaai}
\bibliography{../../../Bib/biblio}

\newpage
\appendix
\section{The Perception-Adjusted Luce Model}\label{App:PALM}

In this section, for the purpose of completeness, we describe the perception-adjusted Luce model (PALM) proposed by in \cite{echenique2018}. The authors recover the role of perception among alternatives using a weak order $\succsim$. The idea is that if $x\succ y$ then $x$ tends to be perceived before $y$, and whenever $x\sim y$ then $x$ and $y$ are perceived at the same time.

A \textit{Perception-Adjusted Luce Model} (PALM) is described by two parameters: a weak order $\succsim$, and a utility function $u$. She perceives elements of an offered set $S\subseteq X$ sequentially according to the equivalence classes induced by $\succsim$ (in our representation, we called them \textit{levels}). Each alternative is selected with probability defined by $\mu$, a function depending on $u$ and closely related to Luce's formula.

\begin{defn}\label{mod:PALM}
	A \emph{Perception-Adjusted Luce Model} (\textbf{PALM}) is a pair $(\succsim,u)$, where the probability of choosing a product $x$ when offering the set $S$ is:
	
	\begin{equation}\label{rho_def}
	\rho (x,S)=\mu(x,S)\cdot\prod_{\alpha \in S/\succsim : \alpha >x}(1-\mu(\alpha,S)),
	\end{equation}
	where:
	\begin{equation}\label{mu_def}
	\mu (x,S)=\frac{u(x)}{\sum_{y \in S}u(y)+u_0}.
	\end{equation}
	
	$S/\succsim$ corresponds to the set of equivalence classes in which $\succsim$ partitions $S$.

	Thus, the probability of choosing a product is the probability of not choosing any products belonging to the previous levels and then selecting the product according a Luce's Model considering all the offered products. We can also write the probability of not choosing any product of the assortment:
	
	\begin{equation}\label{p0_pbb}
	\rho (x_0,S)=1-\sum_{y \in S} \rho (y,S),	
	\end{equation}	
	or equivalently, as the probability of not choosing on each one of the equivalence classes.
	
	\begin{equation}\label{p0_pbb_product}
	\rho (x_0,S)=\prod_{\alpha \in S/\succsim}(1-\mu(\alpha,S))
	\end{equation}	
\end{defn}

In the Perception-Adjusted Luce Model, the customer makes her selection following a sequential procedure. She considers the alternatives in sequence, following a predefined perception priority order. Choosing an alternative is conditioned to not choosing any other alternative perceived before. If none of the offered alternatives is selected, then the outside option is chosen.

Is interesting to note that setting the outside option utility to zero does not result in zero choice probability for the outside option. As explained in \cite{echenique2018}, there are two sources behind choosing the outside option in the PALM. One is the utility of the outside option, which is to the same extent as in Luce's model with an outside option. The second, and to our appreciation the one that differentiate this model from other models of choice, is due the sequential nature of choice. When a customer chooses sequentially following a perception priority order, it can happen that she checks all products in the offered set without making a choice. When this occurs, this seems to increase or bias the value of the outside option probability.

Another consequence of the functional form of the outside option probability (Equation \eqref{p0_pbb_product}), is the ability to model choice overload. In addition, this model allows violations to regularity and violations to stochastic transitivity. The interested reader is referred to \cite{echenique2018} for more details.

\newpage

\section{Appendix: Proofs}
\label{proofs}
\label{appendix}

In this section we provide the proofs missing from the main text.
\printproofs

\end{document}